\pdfoutput=1
\documentclass{llncs}
\usepackage{amsfonts,amsmath}
\usepackage{graphicx}
\usepackage{url}
\usepackage{verbatim}
\usepackage{color}
\definecolor{lgray}{gray}{0.92}
\definecolor{lblue}{rgb}{0.90,0.90,1.00}
\definecolor{lyellow}{rgb}{1.00,1.00,0.70}

\usepackage{listings}
\newenvironment{codex}{\small\verbatim}{\endverbatim\normalsize}
 
\lstnewenvironment{code}
    {\lstset{}%
      \csname lst@SetFirstLabel\endcsname}
    {\csname lst@SaveFirstLabel\endcsname}
\newtheorem{prop}{Proposition}

\newtheorem{df}{Definition}

\newcommand{\BI}[0]{\begin{itemize}}
\newcommand{\EI}[0]{\end{itemize}}

\newcommand{\BE}[0]{\begin{enumerate}}
\newcommand{\EE}[0]{\end{enumerate}}

\newcommand{\BX}[0]{\begin{codex}}
\newcommand{\EX}[0]{\end{codex}}
\newcommand{\BV}[0]{\begin{verbatim}}
\newcommand{\EV}[0]{\end{verbatim}}

\def \bscale1 {0.25}
\def \bscale {0.25}
\def \N {\mathbb{N}}
\def \T {\mathbb{T}}

\newcommand{\FIG}[4]{
\begin{figure}[htbp]
\centering
{\includegraphics[scale=#3]{../figs/#4}}
\caption{#2}
\label{#1}
\end{figure}
}

\title{
  A Prolog Specification of Giant Number Arithmetic
}

\author{Paul Tarau}
\institute{
   Department of Computer Science and Engineering\\
   University of North Texas\\
   {\em e-mail: tarau@cs.unt.edu}
}

\begin{document}
\maketitle
\date{}

\begin{abstract}
The tree based representation described in this paper, {\em hereditarily binary numbers}, applies recursively a run-length compression mechanism
that enables computations limited by the structural
complexity of their operands rather than 
by their bitsizes. While within constant factors
from their traditional counterparts for their
worst case behavior, our arithmetic operations 
open the doors for interesting numerical computations,
impossible with traditional number representations.

We provide a complete specification of our algorithms
in the form of a purely declarative Prolog program.

{\em {\bf Keywords:} 
hereditary numbering systems, 
compressed number representations,
arithmetic computations with giant numbers,
tree-based numbering systems,
Prolog as a specification language.
}

\end{abstract}

\section{Introduction}

While {\em notations} like Knuth's ``up-arrow'' \cite{knuthUp} or 
tetration 
are useful in describing very large numbers, they do not provide 
the ability to actually {\em compute} with them - as, for instance, 
addition or multiplication with a natural number results in a number that
cannot be expressed with the notation anymore.

The novel contribution of this paper is a tree-based numbering
system that {\em allows computations} with numbers comparable in size
with Knuth's ``arrow-up'' notation.
Moreover, these computations have a worst case complexity
that is comparable with the traditional  binary numbers,
while their best case complexity outperforms binary numbers by an arbitrary
tower of exponents factor. Simple operations like successor, 
multiplication by 2, exponent of 2 are practically constant time and
a number of other operations of practical interest 
like addition, subtraction and comparison benefit from
significant complexity reductions.
For the curious reader, it is basically
a {\em hereditary number system} similar to \cite{goodstein},
based on recursively applied {\em run-length} 
compression of a special (bijective) binary digit notation.

A concept of structural complexity is also introduced, based on the
size of our tree representations. It provides estimates
on worst and best cases for our algorithms and it serves as
an indicator of the expected performance of our arithmetic operations.

We have adopted
a {\em literate programming} style, i.e. the
code contained in the paper
forms a self-contained Prolog program 
(tested with SWI-Prolog, Lean Prolog and Styla),
also available as a separate file at
\url{http://logic.cse.unt.edu/tarau/research/2013/hbn.pl} .
We hope that this will encourage the reader to
experiment interactively and validate the
technical correctness of our claims. 

The paper is organized as follows.
Section \ref{bijbin} gives some background on representing
bijective base-2 numbers as iterated function application and
section \ref{herbin} introduces hereditarily binary numbers.
Section \ref{succ} describes practically constant time successor and predecessor operations
on tree-represented numbers.
Section \ref{emu} shows an emulation of bijective base-2 with hereditarily binary numbers
and section \ref{arith} discusses some of their basic arithmetic operations.
Section \ref{stru} defines a concept of structural complexity studies
best and worst cases and  comparisons with bitsizes.
Section \ref{related} discusses related work.
Section \ref{concl} concludes the paper and discusses future work.
Finally, an optimized general multiplication algorithm
is described in the Appendix.

\section{Bijective base-2 numbers as iterated function applications} \label{bijbin}

Natural numbers can be seen
as represented by iterated
applications of the functions $o(x)=2x+1$ and $i(x)=2x+2$
corresponding the so called
{\em bijective base-2} representation \cite{sac12} 
together with the convention
that 0 is represented as the empty sequence. 
As each $n \in \N$ can be seen as a unique composition of these functions
we can make this precise as follows:
\begin{df}
We call bijective base-2 representation of $n \in \N$ the 
unique sequence of applications of functions $o$ and $i$ 
to  $\epsilon$ that evaluates to $n$.
\end{df}
With this representation,
and denoting the empty sequence $\epsilon$, one obtains 
$0=\epsilon, 
1=o(\epsilon),
2=i(\epsilon),
3=o(o(\epsilon)),
4=i(o(\epsilon)),
5=o (i(\epsilon))$ etc. 
and the following holds:
\begin{equation} \label{isucco}
  i(x)=o(x)+1
\end{equation}

\subsection{Properties of the iterated functions $o^n$ and $i^n$}

\begin{prop}\label{fastexp}
Let $f^n$ denote application of function $f$ $n$ times. 
Let $o(x)=2x+1$ and $i(x)=2x+2$, $s(x)=x+1$ and $s'(x)=x-1$. 
Then $k>0 \Rightarrow s(o^n(s'(k))=k2^n$ and 
$k>1 \Rightarrow s(s(i^n(s'(s'(k))))=k2^n$. 
In particular, $s(o^n(0))=2^n$ and $s(s(i^n(0)))=2^{n+1}$.
\end{prop}
\begin{proof}
By induction.
Observe that for $n=0,k>0, s(o^0(s'(k))=k2^0$ because $s(s'(k)))=k$. 
Suppose that $P(n): k>0 \Rightarrow s(o^n(s'(k)))=k2^n$ holds. 
Then, assuming $k>0$, P(n+1) follows, given that
$s(o^{n+1}(s'(k)))=s(o^n(o(s'(k))))=s(o^n(s'(2k)))= 2k2^n= k2^{n+1}$.
Similarly, the second part of the proposition also follows by induction on $n$.
\end{proof}
The underlying arithmetic identities are:

\begin{equation}\label{onk}
o^n(k)=2^n(k+1)-1
\end{equation}
\begin{equation}\label{ink}
i^n(k)=2^n(k+2)-2
\end{equation}
and in particular
\begin{equation}
o^n(0)=2^n-1
\end{equation}\label{onk0}
\begin{equation}\label{ink0}
i^n(0)=2^{n+1}-2
\end{equation}

\section{Hereditarily binary numbers} \label{herbin}

\subsection{Hereditary Number Systems}
Let us observe that conventional number systems, as well as 
the bijective base-2 numeration system described so far, represent 
{\em blocks of contiguous 0 and 1 digits} appearing 
in the binary representation of a number
somewhat naively - one 
digit for each element of the block.
Alternatively, one might think that counting the blocks
and representing the resulting
counters as {\em binary numbers} would be also possible. 
But then, the same principle could be  applied recursively.
So instead of representing each block of 0 or 1 digits by as many
symbols as the size of the block -- essentially a {\em unary} representation --
one could also encode the number of elements in such a block using
a {\em binary} representation.

This brings us to the idea of hereditary number systems.

\subsection{Hereditarily binary numbers as a data type}

First, we define a data type for our tree represented natural numbers,
that we call {\em hereditarily binary numbers} to emphasize that
{\em binary} rather than {\em unary}
encoding is recursively used in their representation.
\begin{df}
The data type $\T$ of the set of hereditarily binary numbers 
is defined inductively as the set of Prolog terms such that:
\end{df}
\begin{equation}
X \in \T~\text{if and only if}~ X=e ~or~X~\text{is of the form}~v(T,Ts)~or~w(T,Ts)
\end{equation}
where $T\in\T$ and 
$Ts$ stands for a finite sequence (list) of elements of $\T$.

The intuition behind the set $\T$ is the following:
\begin{itemize}
\item The term $e$ (empty leaf) corresponds to zero 
\item the term $v(T,Ts)$ counts the number $T+1$ (as counting starts at 0) of {\tt o} applications followed by an {\em alternation}
of similar counts of {\tt i} and {\tt o} applications in $Ts$

\item the term $w(T,Ts)$ counts the number $T+1$ of {\tt i} applications followed by an {\em alternation}
of similar counts of {\tt o} and {\tt i} applications in $Ts$

\item the same principle is applied recursively for the counters, until the empty sequence is reached
\end{itemize}
One can see this process as run-length compressed bijective base-2 numbers, represented as trees
with either empty leaves or at least one branch, after applying the encoding recursively.

\begin{df}
The function $n:\T \to \N$  shown in equation {\bf \ref{deft}} 
defines the unique natural number associated to a term of type $\T$.
\begin{figure*}
\begin{equation}\label{deft}
n(T)=
\begin{cases}
0&  \text{if $T=~${\tt e}},\\
2^{n(X)+1}-1&  \text{if $T=~${\tt v(X,[])}},\\

{(n(U)+1)}2^{n(X)+1}-1&  \text{if $T=~${\tt v(X,[Y|Xs])} and $U=~${\tt w(Y,Xs)}},\\

2^{n(X)+2}-2&  \text{if $T=~${\tt w(X,[])}},\\

{(n(U)+2)}2^{n(X)+1}-2&  \text{if $T=~${\tt w(X,[Y|Xs])} and $U=~${\tt v(Y,Xs)}}.
\end{cases}
\end{equation}
\end{figure*}
\end{df}
For instance, the computation of {\tt N} in 
{\tt ?- n(w(v(e, []), [e, e, e]),N)} expands to
$({({({2^{{0}+1}-1}+2)2^{{0} + 1}-2}+1)2^{{0} + 1}-1}+2)2^{{2^{{0}+1}-1} + 1}-2=42$.
The Prolog equivalent of equation (\ref{deft}) (using bitshifts for exponents of 2) is:
\begin{code}
n(e,0).
n(v(X,[]),R) :-n(X,Z),R is 1<<(1+Z)-1.
n(v(X,[Y|Xs]),R):-n(X,Z),n(w(Y,Xs),K),R is (K+1)*(1<<(1+Z))-1.
n(w(X,[]),R):-n(X,Z),R is 1<<(2+Z)-2.
n(w(X,[Y|Xs]),R):-n(X,Z),n(v(Y,Xs),K),R is (K+2)*(1<<(1+Z))-2. 
\end{code}
The following example illustrates the values associated with
the first few natural numbers.
\begin{codex}
0:e, 1:v(e,[]), 2:w(e,[]), 3:v(v(e,[]),[]), 4:w(e,[e]), 5:v(e,[e])
\end{codex}
Note that a term of the form {\tt v(X,Xs)} represents an
odd number $\in \N^+$ and a term of the form {\tt w(X,Xs)} represents an
even number $\in \N^+$. The following holds:
\begin{prop}
$n:\T \to \N$ is a bijection, i.e., each term canonically represents the
corresponding natural number.
\end{prop}
\begin{proof}
It follows from the identities (\ref{onk}), (\ref{ink}) by 
replacing the power of 2 functions with the
corresponding iterated applications of
$o$ and $i$.
\end{proof}

\section{Successor and predecessor} \label{succ}

We will now specify successor and predecessor
through a {\em reversible} Prolog
predicate  {\tt s(Pred,Succ)} holding
if {\tt Succ} is the successor of {\tt Pred}.
\begin{code}
s(e,v(e,[])).
s(v(e,[]),w(e,[])).
s(v(e,[X|Xs]),w(SX,Xs)):-s(X,SX).
s(v(T,Xs),w(e,[P|Xs])):-s(P,T). 
s(w(T,[]),v(ST,[])):-s(T,ST).
s(w(Z,[e]),v(Z,[e])).
s(w(Z,[e,Y|Ys]),v(Z,[SY|Ys])):-s(Y,SY).
s(w(Z,[X|Xs]),v(Z, [e,SX|Xs])):-s(SX,X). 
\end{code}

It can be proved by structural induction that Peano's axioms hold
and as a result $<\T,e,s>$ is a Peano algebra.

Note that recursive calls to {\tt s} in {\tt s} happen on terms that are
(roughly) logarithmic in the bitsize of their operands.  One can therefore 
assume that their complexity, computed by an {\em iterated logarithm},
is practically 
constant\footnote{Empirically, 
when computing the successor on 
the first $2^{30}=1073741824$ natural 
numbers (with a deterministic, functional equivalent of our reversible {\tt s} and its inverse), 
there are in total
2381889348 calls to {\tt s}, 
averaging to 2.2183 per successor and predecessor computation.
The same average for 100 successor computations on
5000 bit random numbers also oscillates around 2.21.}.
Note also that by using a single reversible predicate {\tt s}
for both successor and predecessor, while the solution is
always unique, some backtracking occurs
in the latest case. One can eliminate this by using two 
specialized predicates for successor and predecessor.

\section{Emulating the bijective base-2 operations {\tt o}, $i$} \label{emu}

To be of any practical interest, we will need to ensure that
our data type $\T$ emulates also binary arithmetic. We will
first show that it does, and next we will show that
on a number of operations like exponent of 2 or multiplication
by an exponent of 2, it significantly lowers complexity.

Intuitively, the first step should be easy, as we need
to express single applications or ``un-applications'' of
{\tt o} and {\tt i} in terms of their iterated applications
encapsulated in the terms of type $\T$.

First we emulate single applications of {\tt o} and {\tt i} seen
in terms of {\tt s}. Note that {\tt o/2} and {\tt i/2}
are also {\em reversible}
predicates.
\begin{code}
o(e,v(e,[])).
o(w(X,Xs),v(e,[X|Xs])).
o(v(X,Xs),v(SX,Xs)):-s(X,SX).

i(e,w(e,[])).
i(v(X,Xs),w(e,[X|Xs])).
i(w(X,Xs),w(SX,Xs)):-s(X,SX).
\end{code}
Finally the ``recognizers'' {\tt o\_} and {\tt i\_}
simply detect {\tt v} and {\tt w} corresponding
to {\tt o} (and respectively {\tt i}) being the last operation applied
and {\tt s\_} detects that the number is a successor, i.e., not
the empty term {\tt e}.
\begin{code}
s_(v(_,_)).    s_(w(_,_)).

o_(v(_,_)).    i_(w(_,_)).
\end{code}

Note that each of the predicates {\tt o} and 
{\tt i} calls {\tt s} and on a term that is (on the average)
logarithmically smaller. As a worst case, the following holds:
\begin{prop}
The costs of {\tt o} and {\tt i} are within a constant factor from  the cost of {\tt s}.
\end{prop}

\begin{df}
The function $t:\N \to \T$ defines the unique tree of type $\T$ associated to a natural number as follows:
\begin{equation}
t(x)=
\begin{cases}
{\tt e}&  \text{if $x=~${\tt 0}},\\
\text{\tt o}(t({{x-1}\over 2})) &  \text{if $x>0$ and $x$ is odd},\\
\text{\tt i}(t({x \over 2}-1)) &  \text{if $x>0$ and $x$ is even}\\
\end{cases}
\end{equation}
\end{df}

We can now define the corresponding Prolog predicate 
that converts from terms of type $\T$ to
natural numbers. Note that we use bitshifts ({\tt >>}) for division by 2.
\begin{code}
t(0,e).
t(X,R):-X>0, X mod 2=:=1,Y is (X-1)>>1, t(Y,A),o(A,R).
t(X,R):-X>0, X mod 2=:=0,Y is (X>>1)-1, t(Y,A),i(A,R).
\end{code}
The following holds: 

\begin{prop}
Let {\tt id} denote $\lambda x.x$ and ``$\circ$'' function composition. Then, on their respective domains
\begin{equation}
t \circ n = id,~~
n \circ t = id
\end{equation}
\end{prop}
\begin{proof}
By induction, using the arithmetic formulas defining the two functions.
\end{proof}
Note also that the cost of $t$ is proportional to the bitsize
of its input and the cost of $n$ is proportional to the 
bitsize of its output.

\section{Arithmetic operations} \label{arith}

\subsection{A few low complexity operations}
Doubling a number {\tt db}
and reversing the {\tt db} operation ({\tt hf}) are
quite simple, once one remembers that the arithmetic
equivalent of function {\tt o} is $\lambda x.2x+1$.
\begin{code}
db(X,Db):-o(X,OX),s(Db,OX).
hf(Db,X):-s(Db,OX),o(X,OX).
\end{code}

Note that efficient implementations follow directly from
our number theoretic observations in section \ref{bijbin}. 
For instance, as a consequence
of proposition \ref{fastexp},
the operation
{\tt exp2} computing an exponent of $2$ ,
has the following simple definition in terms of
{\tt s}.
\begin{code}
exp2(e,v(e,[])).
exp2(X,R):-s(PX,X),s(v(PX,[]),R).
\end{code}

\begin{prop}
The costs of {\tt db, hf} and {\tt exp2} are within a constant factor from the cost of {\tt s}.
\end{prop}
\begin{proof}
It follows by observing that at most 2 calls to {\tt s, o} are
made in each. 
\end{proof}

\subsection{Reduced complexity addition and subtraction}

We now derive  efficient addition and subtraction
operations similar
to the successor/predecessor {\tt s}, that {\em work on one
run-length encoded bloc at a time}, rather than by
individual {\tt o} and {\tt i} steps.

We first
define the predicates {\tt otimes} corresponding to $o^n(k)$ and {\tt itimes} corresponding to $i^n(k)$.
\begin{code}
otimes(e,Y,Y).
otimes(N,e,v(PN,[])):-s(PN,N).
otimes(N,v(Y,Ys),v(S,Ys)):-add(N,Y,S).
otimes(N,w(Y,Ys),v(PN,[Y|Ys])):-s(PN,N).
\end{code}
\begin{code}
itimes(e,Y,Y).
itimes(N,e,w(PN,[])):- s(PN,N).
itimes(N,w(Y,Ys),w(S,Ys)):-add(N,Y,S).
itimes(N,v(Y,Ys),w(PN,[Y|Ys])):-s(PN,N).    
\end{code}
They are part of a chain of {\em mutually recursive predicates} as
they are already referring to the 
{\tt add} predicate, to be implemented later.
Note also that instead of naively iterating, they implement a 
more efficient ``one bloc at a time'' algorithm. 
For instance, when detecting that its argument counts
a number of applications of {\tt o}, {\tt otimes} just increments that count.
On the other hand, when the last predicate applied was {\tt i}, {\tt otimes} simply
inserts a new count for {\tt o} operations. A similar process corresponds to 
{\tt itimes}. As a result, performance is (roughly) logarithmic rather than linear
in terms of the bitsize of argument {\tt N}. We will use this property for
implementing a low complexity multiplication by exponent of 2 operation.

We also need a number of arithmetic identities on $\N$
involving iterated applications of $o$ and $i$.
\begin{prop} \label{addeqs}
The following hold:
\begin{equation} \label{oplus}
o^k(x) + o^k(y) = i^k(x+y) 
\end{equation}
\begin{equation} \label{oiplus}
o^k(x) + i^k(y) = i^k(x)+o^k(y) = i^k(x+y+1)-1 
\end{equation}
\begin{equation} \label{iplus}
i^k(x) + i^k(y) = i^k(x+y+2)-2 
\end{equation}
\end{prop}
\begin{proof}
By (\ref{onk}) and (\ref{ink}), we substitute the $2^k$-based equivalents of $o^k$ and $i^k$,
then observe that the same reduced forms appear on both sides.
\end{proof}

The corresponding Prolog code is:
\begin{code}
oplus(K,X,Y,R):-add(X,Y,S),itimes(K,S,R). 
   
oiplus(K,X,Y,R):-add(X,Y,S),s(S,S1),itimes(K,S1,T),s(R,T).

iplus(K,X,Y,R):-add(X,Y,S),s(S,S1),s(S1,S2),itimes(K,S2,T),s(P,T),s(R,P).
\end{code}
Note that the code uses {\tt add} that we will define later and
that it is part of a chain of mutually recursive predicate calls,
that together will provide an intricate but efficient implementation of
the intuitively simple idea: {\em we want to work on one run-length 
encoded block at a time}.

The corresponding identities for subtraction are:
\begin{prop} \label{subeqs}
\begin{equation}\label{ominus}
x > y ~\Rightarrow~ o^k(x) - o^k(y) = o^k(x-y-1)+1 
\end{equation}
\begin{equation}\label{oiminus}
x>y+1 ~\Rightarrow~ o^k(x) - i^k(y) = o^k(x-y-2)+2
\end{equation}
\begin{equation}\label{iominus}
x \geq y ~\Rightarrow~ i^k(x) - o^k(y) = o^k(x-y)
\end{equation}
\begin{equation}\label{iminus}
x > y ~\Rightarrow~ i^k(x) - i^k(y) = o^k(x-y-1)+1 
\end{equation}
\end{prop}
\begin{proof}
By (\ref{onk}) and (\ref{ink}), we substitute the $2^k$-based equivalents of $o^k$ and $i^k$, 
then observe that the same reduced forms appear on both sides. Note that special cases
are handled separately to ensure that subtraction is defined.
\end{proof}
The Prolog code, also covering the special cases, is:
\begin{code}
ominus(_,X,X,e).
ominus(K,X,Y,R):-sub(X,Y,S1),s(S2,S1),otimes(K,S2,S3),s(S3,R). 
\end{code}
\begin{code}
iminus(_,X,X,e).
iminus(K,X,Y,R):-sub(X,Y,S1),s(S2,S1),otimes(K,S2,S3),s(S3,R). 
\end{code}
\begin{code}
oiminus(_,X,Y,v(e,[])):-s(Y,X).
oiminus(K,X,Y,R):-s(Y,SY),s(SY,X),exp2(K,P),s(P,R).
oiminus(K,X,Y,R):-
  sub(X,Y,S1),s(S2,S1),s(S3,S2),s_(S3), % S3 <> e
  otimes(K,S3,S4),s(S4,S5),s(S5,R).
\end{code}
\begin{code} 
iominus(K,X,Y,R):-sub(X,Y,S),otimes(K,S,R).
\end{code}
Note the use of {\tt sub}, to be defined later, which
is also part of the mutually recursive chain of operations.

The next two predicates extract the iterated applications of
$o^n$ and respectively $i^n$ from {\tt v} and {\tt w} terms:
\begin{code}
osplit(v(X,[]), X,e).
osplit(v(X,[Y|Xs]),X,w(Y,Xs)).
\end{code}
\begin{code}
isplit(w(X,[]), X,e).
isplit(w(X,[Y|Xs]),X,v(Y,Xs)).
\end{code}
We are now ready for defining addition. The base cases are:
\begin{code}
add(e,Y,Y).
add(X,e,X):-s_(X).
\end{code}
In the case when both terms represent odd numbers, we apply with {\tt auxAdd1}
the identity (\ref{oplus}),
after extracting the iterated applications of $o$ as {\tt a} and {\tt b} with the predicate {\tt osplit}.
\begin{code}
add(X,Y,R):-o_(X),o_(Y),osplit(X,A,As),osplit(Y,B,Bs),cmp(A,B,R1),
  auxAdd1(R1,A,As,B,Bs,R). 
\end{code}
In the case when the first term is odd and the 
second even, we apply with {\tt auxAdd2} the identity (\ref{oiplus}),
after extracting the iterated application of $o$ and $i$ as {\tt a} and {\tt b}.
\begin{code}  
add(X,Y,R):-o_(X),i_(Y),osplit(X,A,As),isplit(Y,B,Bs),cmp(A,B,R1),
  auxAdd2(R1,A,As,B,Bs,R).
\end{code}
In the case when the first term is even and the second odd, we apply with {\tt auxAdd3}
the identity (\ref{oiplus}),
after extracting the iterated applications of $i$ and $o$ as, respectively, {\tt a} and {\tt b}.
\begin{code}    
add(X,Y,R):-i_(X),o_(Y),isplit(X,A,As),osplit(Y,B,Bs),cmp(A,B,R1),
  auxAdd3(R1,A,As,B,Bs,R).
\end{code}
In the case when both terms represent even numbers, we apply with {\tt auxAdd4}
the identity (\ref{iplus}),
after extracting the iterated application of $i$ as {\tt a} and {\tt b}.
\begin{code}      
add(X,Y,R):-i_(X),i_(Y),isplit(X,A,As),isplit(Y,B,Bs),cmp(A,B,R1),
  auxAdd4(R1,A,As,B,Bs,R).
\end{code}
Note the presence of the comparison operation {\tt cmp}, to be defined later, also
part of our chain of mutually recursive operations.
Note also that in each case we ensure that a block of the same
size is extracted, depending on which of the two operands {\tt a} or {\tt b} is larger.
Beside that, the auxiliary predicates {\tt auxAdd1, auxAdd2, auxAdd3} and {\tt auxAdd4}
implement the equations of Prop. \ref{addeqs}.
\begin{code}
auxAdd1('=',A,As,_B,Bs,R):- s(A,SA),oplus(SA,As,Bs,R).
auxAdd1('>',A,As,B,Bs,R):-
  s(B,SB),sub(A,B,S),otimes(S,As,R1),oplus(SB,R1,Bs,R).
auxAdd1('<',A,As,B,Bs,R):-
  s(A,SA),sub(B,A,S),otimes(S,Bs,R1),oplus(SA,As,R1,R).  
\end{code}
\begin{code}
auxAdd2('=',A,As,_B,Bs,R):- s(A,SA),oiplus(SA,As,Bs,R).
auxAdd2('>',A,As,B,Bs,R):-
  s(B,SB),sub(A,B,S),otimes(S,As,R1),oiplus(SB,R1,Bs,R).
auxAdd2('<',A,As,B,Bs,R):-
  s(A,SA),sub(B,A,S),itimes(S,Bs,R1),oiplus(SA,As,R1,R).  
\end{code}
\begin{code}  
auxAdd3('=',A,As,_B,Bs,R):- s(A,SA),oiplus(SA,As,Bs,R).
auxAdd3('>',A,As,B,Bs,R):-
  s(B,SB),sub(A,B,S),itimes(S,As,R1),oiplus(SB,R1,Bs,R).
auxAdd3('<',A,As,B,Bs,R):-
  s(A,SA),sub(B,A,S),otimes(S,Bs,R1),oiplus(SA,As,R1,R).    
\end{code}
\begin{code}
auxAdd4('=',A,As,_B,Bs,R):- s(A,SA),iplus(SA,As,Bs,R).
auxAdd4('>',A,As,B,Bs,R):-
  s(B,SB),sub(A,B,S),itimes(S,As,R1),iplus(SB,R1,Bs,R).
auxAdd4('<',A,As,B,Bs,R):-
  s(A,SA),sub(B,A,S),itimes(S,Bs,R1),iplus(SA,As,R1,R).    
\end{code}

The code for the subtraction predicate {\tt sub} is similar:
\begin{code}
sub(X,e,X).
sub(X,Y,R):-o_(X),o_(Y),osplit(X,A,As),osplit(Y,B,Bs),cmp(A,B,R1),
  auxSub1(R1,A,As,B,Bs,R). 
\end{code}
In the case when both terms represent odd numbers, we apply the identity (\ref{ominus}),
after extracting the iterated applications of $o$ as {\tt a} and {\tt b}.
For the other cases, we use, respectively, the identities \ref{oiminus}, \ref{iominus} and
\ref{iminus}:
\begin{code}  
sub(X,Y,R):-o_(X),i_(Y),osplit(X,A,As),isplit(Y,B,Bs),cmp(A,B,R1),
  auxSub2(R1,A,As,B,Bs,R).
\end{code}
\begin{code}  
sub(X,Y,R):-i_(X),o_(Y),isplit(X,A,As),osplit(Y,B,Bs),cmp(A,B,R1),
  auxSub3(R1,A,As,B,Bs,R).
\end{code}
\begin{code} 
sub(X,Y,R):-i_(X),i_(Y),isplit(X,A,As),isplit(Y,B,Bs),cmp(A,B,R1),
  auxSub4(R1,A,As,B,Bs,R).
\end{code}
Note also the auxiliary predicates {\tt auxSub1, auxSub2, auxSub3} and {\tt auxSub4} that
implement the equations of Prop. \ref{subeqs}.

\begin{code}
auxSub1('=',A,As,_B,Bs,R):- s(A,SA),ominus(SA,As,Bs,R).
auxSub1('>',A,As,B,Bs,R):-
  s(B,SB),sub(A,B,S),otimes(S,As,R1),ominus(SB,R1,Bs,R).
auxSub1('<',A,As,B,Bs,R):-
  s(A,SA),sub(B,A,S),otimes(S,Bs,R1),ominus(SA,As,R1,R).  
\end{code}
\begin{code}
auxSub2('=',A,As,_B,Bs,R):- s(A,SA),oiminus(SA,As,Bs,R).
auxSub2('>',A,As,B,Bs,R):-
  s(B,SB),sub(A,B,S),otimes(S,As,R1),oiminus(SB,R1,Bs,R).
auxSub2('<',A,As,B,Bs,R):-
  s(A,SA),sub(B,A,S),itimes(S,Bs,R1),oiminus(SA,As,R1,R).  
\end{code}
\begin{code}  
auxSub3('=',A,As,_B,Bs,R):- s(A,SA),iominus(SA,As,Bs,R).
auxSub3('>',A,As,B,Bs,R):-
  s(B,SB),sub(A,B,S),itimes(S,As,R1),iominus(SB,R1,Bs,R).
auxSub3('<',A,As,B,Bs,R):-
  s(A,SA),sub(B,A,S),otimes(S,Bs,R1),iominus(SA,As,R1,R).    
\end{code}
\begin{code}
auxSub4('=',A,As,_B,Bs,R):- s(A,SA),iminus(SA,As,Bs,R).
auxSub4('>',A,As,B,Bs,R):-
  s(B,SB),sub(A,B,S),itimes(S,As,R1),iminus(SB,R1,Bs,R).
auxSub4('<',A,As,B,Bs,R):-
  s(A,SA),sub(B,A,S),itimes(S,Bs,R1),iminus(SA,As,R1,R).    
\end{code}

\subsection{Defining a total order: comparison}
The comparison operation {\tt cmp}
provides a total order (isomorphic to that on $\N$) on our type $\T$.
It relies on {\tt bitsize} computing the number of applications of
$o$ and $i$ that build a term in $\T$,
which is also part of our mutually recursive predicates, to be defined later.

We first observe that only terms of the same  bitsize
 need detailed comparison,
otherwise the relation between their bitsizes is enough, {\em recursively}.
More precisely, the following holds:
\begin{prop} \label{bitineq}
Let {\tt bitsize} count the number of applications of $o$ or $i$ operations
on a bijective base-2 number. 
Then {\tt bitsize}$(x)~<~${\tt bitsize}$(y) \Rightarrow x<y$.
\end{prop}
\begin{proof}
Observe that, given their lexicographic ordering in ``big digit first'' form, the
bitsize of bijective base-2 numbers is a non-decreasing function.
\end{proof}
\begin{code}
cmp(e,e,'=').
cmp(e,Y,('<')):-s_(Y).
cmp(X,e,('>')):-s_(X).
cmp(X,Y,R):-s_(X),s_(Y),bitsize(X,X1),bitsize(Y,Y1),
  cmp1(X1,Y1,X,Y,R).

cmp1(X1,Y1,_,_,R):- \+(X1=Y1),cmp(X1,Y1,R).
cmp1(X1,X1,X,Y,R):-reversedDual(X,RX),reversedDual(Y,RY),
  compBigFirst(RX,RY,R).
\end{code}
The predicate {\tt compBigFirst} compares two terms known to have the same
{\tt bitsize}. It works on reversed (big digit first) variants,
computed by {\tt reversedDual} and it takes  advantage of the
block structure using the following proposition:
\begin{prop}
Assuming two terms of the same bitsizes, the one starting with
$i$ is larger than one starting with $o$.
\end{prop}
\begin{proof}
Observe that ``big digit first'' numbers are lexicographically ordered with $o<i$.
\end{proof}

As a consequence, {\tt cmp} only recurses when {\em identical} blocks
head the sequence of blocks, otherwise it infers the ``{\tt <}'' or ``{\tt >}''
relation. 
\begin{code}  
compBigFirst(e,e,'=').
compBigFirst(X,Y,R):- o_(X),o_(Y),
  osplit(X,A,C),osplit(Y,B,D),cmp(A,B,R1),
  fcomp1(R1,C,D,R). 
compBigFirst(X,Y,R):-i_(X),i_(Y),
  isplit(X,A,C),isplit(Y,B,D),cmp(A,B,R1),
  fcomp2(R1,C,D,R).
compBigFirst(X,Y,('<')):-o_(X),i_(Y).
compBigFirst(X,Y,('>')):-i_(X),o_(Y).
\end{code}
\begin{code}
fcomp1('=',C,D,R):-compBigFirst(C,D,R).
fcomp1('<',_,_,'>').
fcomp1('>',_,_,'<').   
\end{code}
\begin{code}
fcomp2('=',C,D,R):-compBigFirst(C,D,R).
fcomp2('<',_,_,'<').
fcomp2('>',_,_,'>').
\end{code}
The predicate {\tt reversedDual} reverses the order of application
of the $o$ and $i$ operations to a ``biggest digit
first'' order. For this, it only needs to reverse
the order of the alternative blocks of $o^k$ and $i^k$.
It uses the predicate {\tt len} to compute with {\tt auxRev1} and {\tt auxRev2}
the number of these
blocks. Then, it infers that if the number of blocks is odd, the last block is of the same
kind as the first; otherwise it is of its alternate kind ({\tt w} for {\tt v} and vice versa).
\begin{code}
reversedDual(e,e).
reversedDual(v(X,Xs),R):-reverse([X|Xs],[Y|Ys]),len([X|Xs],L),
  auxRev1(L,Y,Ys,R).
reversedDual(w(X,Xs),R):-reverse([X|Xs],[Y|Ys]),len([X|Xs],L),
  auxRev2(L,Y,Ys,R).
\end{code}
\begin{code}
auxRev1(L,Y,Ys,R):-o_(L),R=v(Y,Ys).
auxRev1(L,Y,Ys,R):-i_(L),R=w(Y,Ys).
\end{code}
\begin{code}  
auxRev2(L,Y,Ys,R):-o_(L),R=w(Y,Ys).
auxRev2(L,Y,Ys,R):-i_(L),R=v(Y,Ys).  
\end{code}
\begin{code}  
len([],e).
len([_|Xs],L):- len(Xs,L1),s(L1,L).
\end{code}

\subsection{Computing {\tt bitsize}}

The predicate {\tt bitsize} computes  the number of applications of the 
{\tt o} and {\tt i} operations. It works by summing up 
the {\em counts} of {\tt o} and {\tt i} operations
composing a tree-represented natural number of type $\T$.
\begin{code}
bitsize(e,e).
bitsize(v(X,Xs),R):-tsum([X|Xs],e,R).
bitsize(w(X,Xs),R):-tsum([X|Xs],e,R).
\end{code}
\begin{code}  
tsum([],S,S).
tsum([X|Xs],S1,S3):-add(S1,X,S),s(S,S2),tsum(Xs,S2,S3).
\end{code}
{\tt Bitsize} concludes our chain of {\em mutually recursive} predicates.
Note that it also provides an efficient implementation
of the integer $log_2$ operation {\tt ilog2}.
\begin{code}
ilog2(X,R):-s(PX,X),bitsize(PX,R).
\end{code}

\subsection{Fast multiplication by an exponent of 2}

The predicate {\tt leftshiftBy} uses Prop. \ref{fastexp}, i.e.,
the fact that repeated application of the {\tt o} operation ({\tt otimes}) 
provides an efficient implementation of 
multiplication with an exponent of 2.
\begin{code}
leftShiftBy(_,e,e).
leftShiftBy(N,K,R):-s(PK,K),otimes(N,PK,M),s(M,R).
\end{code}

The following holds:
\begin{prop}
Assuming {\tt s} constant time,
{\tt leftshiftBy} is (roughly) logarithmic in the bitsize of 
its arguments.
\end{prop}
\begin{proof}
It follows by observing that at most one addition on data
logarithmic in the bitsize of the operands is performed.
\end{proof}

\section{Structural complexity} \label{stru}

As a measure of structural complexity we define
the predicate {\tt tsize} that counts the nodes
of a tree of type $\T$ (except the root).
\begin{code}
tsize(e,e).
tsize(v(X,Xs),R):- tsizes([X|Xs],e,R).
tsize(w(X,Xs),R):- tsizes([X|Xs],e,R).
\end{code}
\begin{code}
tsizes([],S,S).
tsizes([X|Xs],S1,S4):-tsize(X,N),add(S1,N,S2),s(S2,S3),tsizes(Xs,S3,S4).
\end{code}
It corresponds to the function  $c:\T \to \N$
 defined by equation (\ref{tsize}):
\begin{equation} \label{tsize}
c(T)=
\begin{cases}
0&  \text{if $\text{\tt T}=~$e},\\
 \sum_{Y \in \text{\tt [X|Xs]}} {(1+c(Y))} &  \text{if T = \text{\tt v(X,Xs)}},\\
\sum_{Y \in \text{\tt [X|Xs]}} {(1+c(Y))} &  \text{if T = \text{\tt w(X,Xs)}} .
\end{cases}
\end{equation}
The following holds:
\begin{prop} \label{bitcmp}
For all terms $T \in \T$, {\tt tsize(T)} $\leq$ {\tt bitsize(T)}.
\end{prop}
\begin{proof}
By induction on the structure of $T$, by observing that the two
predicates have similar definitions and corresponding calls to
{\tt tsize} return terms assumed smaller than those of
{\tt bitsize}.
\end{proof}

The following example illustrates their use:
\begin{codex}
?- t(123456,T),tsize(T,S1),n(S1,TSize),bitsize(T,S2),n(S2,BSize).
T = w(e, [w(e, [e]), e, v(e, []), e, w(e, []), w(e, [])]),
S1 = w(e, [e, e]), TSize = 12,
S2 = w(e, [w(e, [])]), BSize = 16 .
\end{codex}

worst

After defining the predicate {\tt iterated}, that applies {\tt K} times the predicate {\tt F} 
\begin{code}  
iterated(_,e,X,X).
iterated(F,K,X,R):-s(PK,K),iterated(F,PK,X,R1),call(F,R1,R).
\end{code}
we can exhibit a best case, of minimal structural complexity for its size
\begin{code}
bestCase(K,Best):-iterated(wtree,K,e,Best). 

wtree(X,w(X,[])).
\end{code}
and a worst case, of maximal structural complexity for its size
\begin{code}
worstCase(K,Worst):-iterated(io,K,e,Worst).

io(X,Z):-o(X,Y),i(Y,Z).
\end{code}
The following examples illustrate these predicates:
\begin{codex}
?- t(3,T),bestCase(T,Best),n(Best,N).
T = v(v(e, []), []), Best = w(w(w(e, []), []), []), N = 65534 .

?- t(3,T),worstCase(T,Worst),n(Worst,N).
T = v(v(e, []), []), Worst = w(e, [e, e, e, e, e]), N = 84 .
\end{codex}

It follows from identity (\ref{ink0}) that the predicate {\tt bestCase } computes the iterated
exponent of 2 (tetration) and then applies the predecessor
to it twice, i.e., it computes $2^{2^{\ldots 2}}-2$. 
A simple closed formula (easy to proof by induction) can also be found for {\tt worstCase}: 
\begin{prop}
The predicate {\tt worstCase k} computes the value in $\T$
corresponding to the value ${{4(4^{k} - 1)} \over 3} \in \N$.
\end{prop}

The average space-complexity of our number representation
is related to the average length of the {\em integer partitions of
the bitsize of a number} \cite{part99}. Intuitively, the shorter the
partition in alternative blocks of $o$ and $i$ applications,
the more significant the compression is, but the exact study,
given the recursive application of run-length encoding,
is likely to be quite intricate.

The following example shows that computations with
towers of exponents 20 and 30 levels tall become
possible with our number representation.
\begin{codex}
?- t(20,X),bestCase(X,A),t(30,Y),bestCase(Y,B),add(A,B,C),
|       tsize(C,S),n(S,TSize),write(TSize),nl,fail.
314
\end{codex}
Note that the structural complexity of the result (that we did not print out) is
still quite manageable: {\tt 250}. {\em This opens the door to a new world
where tractability of computations is not limited by the size of the operands 
but only by their structural complexity.}

\section{Related work} \label{related}

Several notations for very large numbers have been invented in the past. Examples
include Knuth's {\em arrow-up} notation \cite{knuthUp}
covering operations like the {\em tetration} (a notation for towers of exponents).
In contrast to our tree-based natural numbers,
such notations are not closed under
addition and multiplication, and consequently
they cannot be used as a replacement
for ordinary binary or decimal numbers.

The first instance of a hereditary number system, at our best knowledge,
occurs in the proof of Goodstein's theorem \cite{goodstein}, where
replacement of finite numbers on a tree's branches by the ordinal $\omega$
allows him to prove that a ``hailstone sequence'' visiting arbitrarily
large numbers eventually turns around and terminates.

Arithmetic packages similar to 
our bijective base-2 view of arithmetic operations
are part of libraries of proof assistants
like Coq \cite{Coq:manual}.

Arithmetic computations based
on recursive data types like
the free magma of binary trees
(isomorphic to the 
context-free language of balanced parentheses)
are described in \cite{sac12},
where they are seen as G\"odel's {\tt System T} types,
as well as combinator application trees.
In \cite{ppdp10tarau}
a type class mechanism is used
to express computations on hereditarily
finite sets and hereditarily finite
sequences.
In \cite{vu09} integer decision diagrams are introduced
providing a compressed representation for sparse
integers, sets and various other data types. 

\section{Conclusion and future work} \label{concl}

We have shown that {\em computations} like
addition, subtraction, exponent of 2 and
 bitsize can be performed with
giant numbers in constant time or time
proportional to their structural complexity rather
than their bitsize.
As {\em structural
complexity} is bounded by bitsize,
our computations are within constant time from
their traditional counterparts, as also
illustrated by our best and worst case
complexity cases.

The fundamental theoretical challenge raised at this point is the following:
{\em can more number-theoretically interesting operations
expressed succinctly in terms of our tree-based data type? Is it possible to reduce the
complexity of some other important operations, besides those found so far?}

The general multiplication algorithm in the {\tt Appendix} shows a first step in 
that direction.


\section*{Appendix}

\subsection*{Reduced complexity general multiplication}

We can devise a similar optimization as for {\tt add} and {\tt sub} for multiplication
\begin{prop}
The following holds:
\begin{equation} \label{muleq}
o^n(a)o^m(b)=o^{n+m}(ab+a+b)-o^n(a)-o^m(b)
\end{equation}
\end{prop}
\begin{proof}
By \ref{onk}, we can expand and then reduce as follows:
$o^n(a)o^m(b) = 
(2^n(a+1)-1)(2^m(b+1)-1)=
2^{n+m}(a+1)(b+1)-(2^n(a+1)+2^m(b+1)+1=
2^{n+m}(a+1)(b+1)-1-(2^n(a+1)-1+2^m(b+1)-1+2)+2=
o^{n+m}(ab+a+b+ 1)-(o^n(a)+o^m(b))-2+2=
o^{n+m}(ab+a+b)-o^n(a)-o^m(b)$
\end{proof}

The corresponding Prolog code starts with the obvious base cases:
\begin{code}
mul(_,e,e).
mul(e,Y,e):-s_(Y).
\end{code}
When both terms represent odd numbers we apply the identity (\ref{muleq}):
\begin{code}
mul(X,Y,R):-o_(X),o_(Y),osplit(X,N,A),osplit(Y,M,B),
  add(A,B,S),mul(A,B,P),add(S,P,P1),s(N,SN),s(M,SM),
  add(SN,SM,K),otimes(K,P1,P2),sub(P2,X,R1),sub(R1,Y,R).
\end{code}
The other cases are reduced to the previous one by the identity
$i=s \circ o$.
\begin{code}  
mul(X,Y,R):-o_(X),i_(Y),s(PY,Y),mul(X,PY,Z),add(X,Z,R).
mul(X,Y,R):-i_(X),o_(Y),s(PX,X),mul(PX,Y,Z),add(Y,Z,R).
mul(X,Y,R):-i_(X),i_(Y),
  s(PX,X),s(PY,Y),add(PX,PY,S),mul(PX,PY,P),add(S,P,R1),s(R1,R).
\end{code}
Note that when the operands are composed of large blocks of alternating
$o^n$ and $i^m$ applications, the algorithm works (roughly) in time proportional 
to the number of blocks rather than the number of digits.  
The following example illustrates a multiplication with two ``tower of exponent'' terms:
\begin{codex}
?- t(30,X),bestCase(X,A),s(A,N),t(40,Y),bestCase(Y,B),s(B,M),
   mul(M,N,P),tsize(P,S),n(S,TSize),write(TSize),nl,fail.
668
\end{codex}
The structural complexity of the result, {\tt 668} is in indicator that such computations are
still tractable. Note however, that the predicate {\tt mul} can be still optimized, by using in its 
last 3 clauses identities similar to (\ref{muleq}), so that it works in all cases one block at a time and it reduces to addition / subtraction operations proportional to the number of blocks.

\begin{codeh}


add0(e,Y,Y).
add0(v(X,Xs),e,v(X,Xs)).
add0(w(X,Xs),e,w(X,Xs)).
add0(X,Y,Z):-o(X1,X),o(Y1,Y),add0(X1,Y1,R),i(R,Z).
add0(X,Y,Z):-o(X1,X),i(Y1,Y),add0(X1,Y1,R),s(R,S),o(S,Z).
add0(X,Y,Z):-i(X1,X),o(Y1,Y),add0(X1,Y1,R),s(R,S),o(S,Z).
add0(X,Y,Z):-i(X1,X),i(Y1,Y),add0(X1,Y1,R),s(R,S),i(S,Z).

sub0(X,e,X).
sub0(X,Y,Z):-o(X1,X),o(Y1,Y),sub0(X1,Y1,R),o(R,R1),s(Z,R1).
sub0(X,Y,Z):-o(X1,X),i(Y1,Y),sub0(X1,Y1,R),o(R,R1),s(R2,R1),s(Z,R2).
sub0(X,Y,Z):-i(X1,X),o(Y1,Y),sub0(X1,Y1,R),o(R,Z).
sub0(X,Y,Z):-i(X1,X),i(Y1,Y),sub0(X1,Y1,R),o(R,R1),s(Z,R1).


tl(N,R):-o_(N),o(R,N). 
tl(N,R):-i_(N),s(v(_,Xs),N),ftl(Xs,R).
 
ftl([],e).
ftl([Y|Ys],R):-
  i(P,w(Y,Ys)),
  s(P,R).

syracuse(N,R):- 
 i(N,M),
 add(N,M,S),
 tl(S,R).
 

nsyr(X,N:R):-nsyr(X,R,0,N).

nsyr(X,R,N1,N2):-
  syracuse(X,S),
  !,
  nsyr1(S,R,N1,N2).

nsyr1(e,S,N1,N2):-!,S=e,N2=N1.
nsyr1(S,S,N,N).
nsyr1(S,R,N1,N3):-N2 is N1+1,nsyr(S,R,N2,N3).

nsyrs(e,[e]):-!.
nsyrs(N,[N|Ns]):-
  syracuse(N,S),
  nsyrs(S,Ns).

op1(F,A,C):-t(A,N),call(F,N,R),n(R,C).  
op2(F,A,B,C):-t(A,N),t(B,K),call(F,N,K,R),n(R,C).  
  
t0:-
   between(0,10,I),
     t(I,X),n(X,N),
     write([I=N,X]),nl,
   fail.
   
 t1:-
   between(0,20,I),
     t(I,X),n(X,N),
     s(P,X),
     write([I=N,X+P]),nl,
   fail.

t2:-
  member(A,[0,1,20,33,100]),member(B,[0,10,11]),
  t(A,X),t(B,Y),add(X,Y,Z),n(Z,R),
  write(A+B=R),nl,
  (A+B=:=R->true;write(error(A+B=R))),nl,
  fail.

t3:-
  member(A,[100,121,133]),member(B,[10,52,100]),
  t(A,X),t(B,Y),sub(X,Y,Z),n(Z,R),
  write(A-B=R),nl,
  (A-B=:=R->true;write(error(A-B=R))),nl,
  fail.

fermat(N,R):-exp2(N,N1),exp2(N1,N2),s(N2,R).

mersenne(P,R):-exp2(P,P1),s(R,P1).

perfect(P,R):-s(P1,P),s(Q,P1),s(v(Q,[Q]),R).

prime48(57885161).

mersenne48(R):-
  prime48(P),
  t(P,T),
  mersenne(T,R).
 
mt:-mersenne48(X),nsyr(X,T:_),write(T),nl,fail.
\end{codeh}

\begin{codeh}
go(M):-
  N is 2^M-1,
  t(N,T),
  time((go1(T))),
  fail.
  
go1(e):-!.
go1(T):-s(PT,T),go1(PT).  
\end{codeh}

\end{document}